\documentclass[reqno,11pt,a4paper]{amsart}
\usepackage{csquotes}
\usepackage{amssymb}
\usepackage{mleftright}  
\usepackage{xfrac} 
\usepackage{mathrsfs}
\usepackage[dvipsnames]{xcolor}
\usepackage[colorlinks,linkcolor=RedViolet,citecolor=PineGreen,anchorcolor=Green,urlcolor=RedViolet]{hyperref}
\usepackage[headings]{fullpage}
\usepackage[numbers,square]{natbib}
\usepackage[shortlabels,inline]{enumitem}
\usepackage{anyfontsize}
\usepackage{scalerel}
\usepackage{bbm} 
\usepackage[foot]{amsaddr}
\usepackage[hang,flushmargin]{footmisc} 
\usepackage{comment}

\newcommand{\adjustedaccent}[1]{%
  \mathchoice{}{}
    {\mbox{\raisebox{-.5ex}[0pt][0pt]{$\scriptstyle#1$}}}
    {\mbox{\raisebox{-.35ex}[0pt][0pt]{$\scriptscriptstyle#1$}}}
}
\newcommand\invbreve[1]{\overset{\adjustedaccent{\smallfrown}}{#1}}
\makeatletter
\newcommand{\optionaldesc}[2]{%
  \phantomsection
  #1\protected@edef\@currentlabel{#1}\label{#2}%
}
\makeatother

\newcommand{\textcite}{\citet*}
\newcommand{\E}{\textsf{\upshape E}}
\newcommand{\Var}{\textsf{\upshape Var}}

\renewcommand{\P}{\textsf{\upshape P}}
\newcommand{\Q}{\textsf{\upshape Q}}

\newcommand{\indicator}[1]{\mathbf{1}_{#1}}
\DeclareMathOperator*{\esssup}{ess\,sup}
\DeclareMathOperator*{\essinf}{ess\,inf}
\newcommand{\newinf}{\mathop{\mathrm{inf}\vphantom{\mathrm{sup}}}}

\newcommand{\vfrac}[2]{\ensuremath{#1 / #2}}

\newcommand{\N}{\mathbb{N}}

\newcommand{\R}{\mathbb{R}}

\newcommand{\id}{{\operatorname{id}}}

\newcommand{\aux}{p}

\newcommand{\hW}{\hat{W}}
\newcommand{\tW}{\widetilde{W}}
\newcommand{\tK}{\widetilde{K}}
\newcommand{\MV}{\textsc{mv}}
\newcommand{\MMV}{\textsc{mmv}}
\newcommand{\M}{\textsc{m}}

\newcommand{\xinf}{\underline{x\mkern-1.5mu}\mkern1.5mu}
\newcommand{\xsup}{\overline{x}}
\newcommand{\xmid}{\hat{x}}
\newcommand{\bliss}{\invbreve{x}}
\renewcommand{\d}{\mathrm{d}}

\newcommand{\e}{\mathrm{e}}

\DeclareMathOperator{\dom}{dom}

\newcommand\shorturl[1]{%
  \href{http://#1}{\nolinkurl{#1}}%
}

\makeatletter
\newcommand*{\bigcdot}{}
\DeclareRobustCommand*{\bigcdot}{%
  \mathbin{\mathpalette\bigcdot@{}}%
}
\newcommand*{\bigcdot@scalefactor}{.7}
\newcommand*{\bigcdot@widthfactor}{1.15}
\newcommand*{\bigcdot@}[2]{%
  \sbox0{$#1\vcenter{}$}
  \sbox2{$#1\cdot\m@th$}%
  \hbox to \bigcdot@widthfactor\wd2{%
    \hfil
    \raise\ht0\hbox{%
      \scalebox{\bigcdot@scalefactor}{%
        \lower\ht0\hbox{$#1\bullet\m@th$}%
      }%
    }%
    \hfil
  }%
}
\makeatother

\newtheorem{theorem}{Theorem}[section]
\theoremstyle{plain}

\newtheorem{corollary}[theorem]{Corollary}

\newtheorem{example}[theorem]{Example}

\newtheorem{proposition}[theorem]{Proposition}
\newtheorem{remark}[theorem]{Remark}

\numberwithin{equation}{section}

\begin{document}
\title[{}]{Cash-invariant hull representation\\of divergence preferences}
\author{Ale\v{s} \v{C}ern\'{y}$^{a,b}$}
\address{$^a$Bayes Business School, City St George's, University of London\\
106 Bunhill Row, London EC1Y 8TZ, UK\vspace{-1.1ex}}
\email{\!ales.cerny@citystgeorges.ac.uk}
\author{Johannes Ruf$^b$}
\address{$^b$Department of Mathematics\\
  London School of Economics and Political Science\\
	Houghton Street, London, WC2A 2AE, UK\vspace{-1.1ex}}
\email{\!j.ruf@lse.ac.uk}
\author{Martin Schweizer$^c$}
\address{$^c$Department of Mathematics\\
 ETH Zurich\\
	R\"amistrasse 101, CH-8092 Zurich, Switzerland\vspace{0.4ex}}
\email{\!martin.schweizer@math.ethz.ch}
\date{\today}

\begin{abstract}
Uniformly weighted divergence preferences (UWDP) introduced in \citet{maccheroni.al.06} are an important class of risk-averse preferences that contain as a special case the monotone mean--variance utility. UWDP are characterised by the lowest expected value of an act in $L^\infty$ under an adversarially chosen probability measure combined with the divergence of this measure. Our main result provides an alternative, computationally friendlier formula, which establishes in full generality that UWDP are the translation-invariant hull of state-independent expected utility over $L^0$. Some consequences of the new representation are studied.   
\end{abstract}

\maketitle
\mleftright 
\section{Introduction}
This paper continues the study of divergence preferences introduced and axiomatized in \textcite{maccheroni.al.06} as part of the wider class of so-called variational preferences. We restrict our attention to uniformly weighted divergence preferences, which are known to be probabilistically sophisticated \cite[Theorem~16]{maccheroni.al.06} and already cover the special instance of the monotone mean--variance (MMV) preferences of \textcite{maccheroni.al.09}. In this case, for a suitable non-negative convex function $f$ with $f(1)=0$, the wealth distributions (acts) $W$ are ranked by the divergence utility $V:L^\infty\to \R$, given by 
\begin{equation}\label{eq:V}
V(W) = \inf_{ Z\in L_{+}^{1}, \E[Z]= 1 }\{\E[WZ+f(Z)]\}.
\end{equation}
In \eqref{eq:V}, the act $W\in L^\infty$ is evaluated based on the worst outcome of its expected value $\E[WZ]$ under alternative models of the world ``penalized'' by the divergence $\E[f(Z)]$. It follows from \eqref{eq:V} that divergence preferences are \emph{cash invariant}, i.e., for $\eta\in\R$ and $W\in L^\infty$, one has $V(W+\eta)=V(W)+\eta$. 

In practice, for example when studying optimal portfolio allocation, it is important to extend the domain of UWDP utility beyond $L^\infty$.  To gain a better understanding of \eqref{eq:V} and its possible extensions, let us consider the specific   divergence function 
\begin{equation*}
f_{\MV} = \frac{1}{2}(\id-1)^2
\end{equation*}
as an example.  In such case, the domain of the divergence utility can be extended to $L^2$ by setting 
\begin{equation}\label{eq:MMV1}
V_{\MMV}(W) ={}  \inf_{ Z\in L^{2}_+, \E[Z]= 1 }\Big\{\E\Big[WZ+\frac{1}{2}(Z-1)^2\Big]\Big\},\quad W\in L^2.
\end{equation}
Note that this expression agrees with \eqref{eq:V} if $W \in L^\infty$ and $f = f_{\MV}$.

Let us now provide a useful alternative representation of $V_{\MMV}$. To this end, we remove the positivity constraint on $Z$ in \eqref{eq:MMV1} and define the standard mean-variance utility $V_{\MV}: L^2 \to \R$ by
\begin{align}\label{eq:MV}
V_{\MV}(W) ={}& 
\displaystyle\inf_{ Z\in L^{2}, \E[Z]= 1 }\Big\{\E\Big[WZ+f_{\MV}(Z)\Big]\Big\}
= \E[W]-\frac{1}{2}\Var(W). 
\end{align}
Indeed, the second equality  follows by utilising the identity 
$$\E[WZ+\tfrac{1}{2}(Z-1)^2]=\E[W+(Z-1)(W-\E[W])+\tfrac{1}{2}(Z-1)^2]$$
for $Z\in L^{2}$ with $\E[Z]= 1$, completing the square, and observing that 
$$
\inf_{ Z\in L^{2}, \E[Z]= 1 }\left\{\E\left[(Z-1+W-\E[W])^2\right]\right\}=0
$$
for any $W \in L^2$.

\textcite[Theorem~24]{maccheroni.al.06} and  \textcite[Theorem~2.1]{maccheroni.al.09} show that $V_{\MMV}$ is the minimal monotone modification of $V_{\MV}$, i.e., 
$$V_{\MMV}(W) = \sup_{Y\in L^\infty_+} V_{\MV}(W-Y).$$  
For this reason, $V_{\MMV}(W)$ is called the \emph{monotone mean--variance} utility.
The point of departure of this paper is an alternative characterisation of $V_{\MMV}$ provided in \textcite{cerny.20}. This begins with an observation that the classical mean--variance utility can also be obtained as the cash-invariant hull of the expected quadratic utility 
\begin{equation*}
g_{\MV} =  \id - \tfrac{1}{2} \id^2,
\end{equation*} 
i.e., 
\begin{equation}\label{eq:MV2}
V_{\MV}(W) ={}  \sup_{ \eta\in \R }\, \left\{\E\left[g_{\MV}(W+\eta)\right]-\eta\right\},\quad W \in L^2.
\end{equation}
Since the operations of taking the cash-invariant hull and the monotone hull commute, it is seen \cite[Equation (9)]{cerny.20} that $V_{\MMV}$ is the cash-invariant hull of the expected monotonized quadratic utility 
\begin{equation*}
g_{\MMV} = \id\wedge1-\tfrac{1}{2}(\id\wedge1)^2,
\end{equation*} 
i.e.,
\begin{equation}\label{eq:MMV2}
V_{\MMV}(W) ={}  \sup_{ \eta\in \R }\, \left\{\E\left[g_{\MMV}(W+\eta)\right]-\eta\right\},\quad W \in L^2.
\end{equation}

Observe that $f=f_{\MV}$ and $g = g_{\MV}$ stand in a conjugate relationship,
\begin{equation}\label{concave conjugate}
g(x) = \inf_{y\in\R}\{xy+f(y)\},\qquad x\in\R.
\end{equation} 
The evident equality between \eqref{eq:MV} and \eqref{eq:MV2}, and the somewhat unexpected equality between \eqref{eq:MMV1} and \eqref{eq:MMV2} suggest the more general identity
\begin{equation}\label{OCE}
\sup_{\eta\in\R}\,\{\E[g(W+\eta)]-\eta\} = \inf_{Z\in L^{1}, \E[Z]= 1 } \{\E[WZ+f(Z)] \}
\end{equation}
for conjugate pairs $f,g$ and \emph{all} appropriately chosen random variables $W$. Note that the monotone mean--variance divergence  $f_{\MMV}$ now explicitly penalizes negative values, i.e., one must take $f_{\MMV} = f_{\MV} + \infty\indicator{\id<0}$ to obtain that $g = g_{\MMV}$ and $f = f_{\MMV}$ are conjugate as per \eqref{concave conjugate}.

For $W\in L^\infty$, the identity \eqref{OCE} was established in \textcite[Theorem~4.2]{bental.teboulle.07} for $f$ finite only inside the positive half-line and with $1$ in the interior of its domain, i.e., for monotone $g$ whose slopes take values on both sides of $1$. The next theorem extends \eqref{OCE} to random variables $W$ from the widest possible subset of $L^0$ (e.g., the domain of \eqref{eq:MMV1} and \eqref{eq:MMV2} is $L^0_+-L^2_+$) and places no restrictions on $f$ other than convexity and lower semicontinuity. The theorem establishes that there is no conceptual gap between divergence preferences on the one hand and the translation-invariant hull of expected utility on the other hand.
\begin{theorem}\label{T1}
 Assume $f:\R\to (-\infty,\infty]$ is a convex lower semicontinuous function. Denote by $g$ the concave conjugate of $-f$. Consider $W\in L^0$ and assume that there exists $\theta\in\R$ such that $g(W+\theta)^-\in L^1$. Then $(WZ + f(Z))^- \in L^1$ for any $Z \in L^1$ and
\begin{equation}\label{eq:T1}
\sup_{\eta\in\R}\,\{\E[g(W+\eta)]-\eta\} 
= \inf_{ Z\in L^{1}, \E[Z]= 1 } \{\E[WZ+f(Z)]\},
\end{equation}
where we set $\E[g(W+\eta)] = -\infty$, whenever $g(W+\eta)^-\notin L^1$. Furthermore, if $1$ is in the interior of $\dom f$, then the supremum in \eqref{eq:T1} is attained.
\end{theorem}
The paper is organised as follows. Section~\ref{S:2} is devoted to the proof of the main theorem, related statements, and examples. Section~\ref{S:3} discusses three consequences of the cash-invariant representation: in Subsection~\ref{SS:3.1} we study the maximal complete market divergence utility, in Subsection~\ref{SS:3.2} we discuss canonical form of UWDP, and in Subsection~\ref{SS:3.3} we highlight useful lessons for domains of monotonicity of non-monotone UWDP, such as the classical mean--variance utility~\eqref{eq:MV2}.  

\section{Proof of the main theorem, related statements, and examples}\label{S:2}

\begin{proof}[Proof of Theorem~\ref{T1}]  
If $f =\infty$, then $g=\infty$ and the statement holds trivially. Assume now $\dom f$ is not empty. Observe that $g$ is then proper, concave, and upper semicontinuous by the standard properties of conjugates; see \textcite[Theorem~12.2]{rockafellar.70}. When $\dom g$ is a singleton, say $\{c\}$, we have that $f = g(c)-c\mkern2mu\id$, $W$ is constant, and \eqref{eq:T1} therefore holds easily. We may henceforth assume 
\begin{equation}\label{eq:dom g}
\xinf = \inf\dom g < \sup\dom g = \xsup.
\end{equation} 

By \cite[Corollary~23.5.1]{rockafellar.70}, $\partial g$ takes values in $\dom f$ and by \cite[Theorem~24.1]{rockafellar.70} one has
\begin{equation}\label{eq:240605}
\inf\dom f = g'_-(\infty),\qquad \sup\dom f = g'_+(-\infty),
\end{equation}
with the convention that $g'_+(-\infty)=\infty$ if $\dom g$ is bounded below and $g'_-(\infty)=-\infty$ if $\dom g$ is bounded above. From \eqref{concave conjugate}, we obtain the Fenchel inequality 
\begin{equation}\label{Fenchel ineq}
g(x) \leq xy+f(y),\qquad x,y\in\R,
\end{equation}
which yields 
\begin{equation} \label{eq:240209}
g(W + \theta) \leq (W+\theta)Z + f(Z)
\end{equation}
 for any $Z\in L^0$. By assumption, there is $\theta\in\R$ such that $g(W+\theta)^-\in L^1$; hence we also have $(WZ + f(Z))^- \in L^1$ for any $Z \in L^1$. 
If $g(W + \theta)\notin L^1$, then the supremum of $+\infty$ is attained at $\theta$ and the two sides of \eqref{eq:T1} are equal by the Fenchel inequality \eqref{eq:240209}. In addition to \eqref{eq:dom g}, we may therefore from now on assume that $g(W + \theta)\in L^1$.

Consider first the case when $1$ does not lie in the interior of the domain of $f$. Then by \eqref{eq:240605} one has either $g'_+(-\infty) \leq 1$ or $g'_-(\infty) \geq 1$. Therefore, the function $g(W+\id)-\id$ is monotone. 
By the Fenchel--Moreau theorem (e.g., \cite[Theorem~12.2]{rockafellar.70}), we have
$$\E\left[\sup_{\eta\in\R}\{g(W+\eta)-\eta\}\right] = \E[W + f(1)] =  \inf_{ Z\in L^{1}, \E[Z]= 1 } \{\E[WZ+f(Z)]\}$$
since any $Z$ not almost surely equal to $1$ implies $\E[f(Z)] = \infty$.  By the monotonicity of the function $g(W+\id)-\id$ and by monotone convergence, in conjunction with the assumption that $g(W+\theta)^-\in L^1$,   we may exchange the expectation and the supremum on the left-hand side of the last display. This yields the statement in the case when $1$ does not lie in the interior of the domain of $f$. 

From now on we assume that $1$ is an interior point of $\dom f$, i.e. by \eqref{eq:240605},
\begin{equation*}
g'_-(\infty)<1<g'_+(-\infty).
\end{equation*}  
The function $\aux:\R\to[-\infty,\infty)$ given by
\begin{equation*}
\aux(\eta)=\E[g(W+\eta )]-\eta
\end{equation*}
is  concave and proper thanks to $g(W + \theta)\in L^1$. We claim $g(W+\eta)^+\in L^1$ for all $\eta\in\R$. For non-monotone $g$ this holds trivially since $g$ is bounded above.  For monotone $g$, one has
$$g(W+\eta)^+\leq g(W+\theta)^{+} + g(\eta-\theta)^{+} + g(\theta-\eta)^{+},\qquad \eta,\theta\in\R,$$ 
and the claim follows.
The function $\aux$ is also upper semicontinuous by Fatou. By concavity and monotone convergence, on the interior of its domain, the one-sided derivatives read
$\aux'_+(\eta) = \E[g'_{+}(W+\eta )]-1$ and $\aux'_-(\eta) = \E[g'_{-}(W+\eta )]-1$. We claim $\aux$ has an optimizer. Indeed, this follows immediately if $\dom g$ is bounded. On unbounded portions of the domain (if any), one has
$\lim_{\eta\uparrow\infty}\aux'_-(\eta)=g'_-(\infty)-1<0$ and $\lim_{\eta\downarrow-\infty}\aux'_+(\eta)=g'_+(-\infty)-1>0$, which yields the claim. Therefore, there is an $\hat\eta\in\R$ such that
\[
\sup_{\eta\in\R}\,\{\E[g(W+\eta)]-\eta\}=\E[g(W+\hat\eta )]-\hat\eta.
\]%

We shall prove 
\begin{equation}\label{eq:240217}
\E[g(\hW)] = \inf_{ Z\in L^{1}, \E[Z]= 1 } \{\E[\hW Z+f(Z)]\},
\end{equation}
where we set $\hW = W + \hat\eta$, which after rearranging yields \eqref{eq:T1}. To this end, one may distinguish four subcases depending on the location of $\hat{\eta}$.
\begin{enumerate}[(A)] 
\item\label{case.A} $\hat{\eta}$ is an interior maximum of $\aux$;
\item\label{case.B} $\hat{\eta}=\inf\dom\aux<\sup\dom\aux$, yielding 
$-\infty<\aux'_{+}(\hat{\eta})\leq 0$;
\item\label{case.C} $\hat{\eta}=\sup\dom\aux>\inf\dom\aux$, yielding 
 $\infty>\aux'_{-}(\hat{\eta})\geq 0$;
\item\label{case.D} $\hat{\eta}$ is the sole element of $\dom \aux$.
\end{enumerate}

The Fenchel inequality \eqref{Fenchel ineq} holds with equality for any $y\in\partial g(x)$ at each point $x$ where $g$ is subdifferentiable; see \cite[Theorem~23.5]{rockafellar.70}. Hence,
\begin{equation}\label{eq:221216}
g(\hW)=\hW Z + f(Z)\qquad \text{for any $Z$ with values in $\partial g(\hW)$.}
\end{equation}%

In case~\ref{case.A}, the optimality condition $\infty> \aux'_-(\hat{\eta})\geq 0\geq \aux'_+(\hat{\eta})> -\infty$ gives
$$ \infty>\E[g'_{-}(\hW)]\geq 1\geq\E[g'_{+}(\hW)]> -\infty.$$
A suitable convex combination of $g'_{+}(\hW)$ and $g'_{-}(\hW)$ thus yields $Z\in L^1$ with values in $\partial g(\hW)$ and $\E[Z]=1$, which together with \eqref{eq:221216} and the Fenchel inequality establishes \eqref{eq:240217} in case~\ref{case.A}.
 
In case~\ref{case.B},  we will show that there is a sequence $(Z_n)_{n\in\N}$  in $L^1$ with $\E[Z_n]=1$ such that
\begin{equation}\label{eq:1}
\E[g(\hW)] = \lim_{n \uparrow \infty}  \E[\hW Z_n +f(Z_n)].
\end{equation}%
Indeed, since $\inf \dom \aux=\hat{\eta}$, we have $g(\hW)\in L^1$ but 
$g(\hW-\frac{1}{2n})\notin L^1$ for each $n\in\N$. Proposition~\ref{P:240424} with $s =1 -\E[g'_+(\hW)] =-p'_+(\hat\eta)\geq0$ 
now yields, for each $n\in\N$, a random variable $Z_n$ with $\E[Z_n]=1$ such that 
\begin{equation*}
0\leq \E[\hW Z_n+f(Z_n)]-\E[g(\hW)]\leq (s+\E[|g'_+(\hW)|])\frac{1}{n}.
\end{equation*}
This completes the proof of part \ref{case.B}. For the economic significance of this construction interpreted as a downward shift in wealth, see Remark~\ref{R:240617}.

The proof of case \ref{case.C} is a mirror image of case \ref{case.B}, i.e., one has $\E[g'_-(\hW)]-1\geq 0$ instead of $1-\E[g'_+(\hW)]\geq 0$. 
One proceeds as in case \ref{case.B} but with Corollary~\ref{C:240624} replacing Proposition ~\ref{P:240424}. Intuitively, one performs an upward shift in wealth to reduce the marginal utility, see Remark~\ref{R:240617}. 

Case \ref{case.D}: Since $\dom p = \{\hat{\eta}\}$, one has $\essinf \hW = \xinf$, $\esssup \hW = \xsup$, $g'_+(-\infty)=\infty$, and $g'_-(\infty)=-\infty$. In view of this, $g(\hW)\in L^1$ yields $\hW\in L^1$. Recalling \eqref{eq:dom g}, fix $\xmid\in (\xinf,\xsup)$ and let 
$$ A = \{\hW<\xmid\}.$$
Note that $\P[A] > 0$.

Consider first the case  $g'_+(\hW)\indicator{A}\in L^1$ and $g'_-(\hW)\indicator{A^c}\in L^1$ and define
$$ s = 1-\E[g'_+(\hW)\indicator{A}]-\E[g'_-(\hW)\indicator{A^c}].$$
For $s\geq 0$, Proposition~\ref{P:240424} with $\varepsilon = \frac{1}{n}$ yields a random variable $Y_n$ such that $\E[Y_n\indicator{A}]=s+\E[g'_+(\hW)\indicator{A}]$ and
$$ 0\leq \E[(\hW Y_n+f(Y_n)-g(\hW))\indicator{A}]\leq (s+\E[|g'_+(\hW)|\indicator{A}])\frac{1}{n}.$$ 
Taking $Z_n = Y_n\indicator{A}+g'_-(\hW)\indicator{A^c}$, one has $\E[Z_n]=1$ and \eqref{eq:1} holds thanks to the Fenchel equality $\hW Z_n +f(Z_n) = g(\hW)$ on $A^c$. For $s<0$, one proceeds analogously with Corollary~\ref{C:240624} replacing Proposition ~\ref{P:240424}.

Suppose now $g'_+(\hW)\indicator{A}\in L^1$ but $g'_-(\hW)\indicator{A^c}\notin L^1$. The latter yields that $g$ is decreasing near $\xsup$. Without loss of generality we may therefore suppose that $g$ is decreasing on $A^c$. Consider a nondecreasing sequence 
$(\xsup_n)_{n \in \N}$ with 
$\xmid<\xsup_n$ and $\lim_{n\uparrow \infty} x_n = \xsup$. 
Since  
$$s_n=  1-\E[g'_+(\hW)\indicator{A}+g'_-(\hW\wedge \xsup_n)\indicator{A^c}] \uparrow\infty, \qquad n\uparrow \infty,$$ 
there is $N\in\N$ such that $s_n>0$ for all $n>N$. For each $n>N$, Proposition~\ref{P:240424} applied on the event $A=\{\hW<\xmid\}$ with $\varepsilon = \frac{1}{n(s_n\vee1)}$ yields a random variable $Y_n$ such that $\E[Y_n\indicator{A}]=s_n+\E[g'_+(\hW)\indicator{A}]$ and
$$ 0\leq \E[(\hW Y_n+f(Y_n)-g(\hW))\indicator{A}]\leq (s_n+\E[|g'_+(\hW)|\indicator{A}])\frac{1}{n(s_n\vee1)}\downarrow 0, \qquad n\uparrow \infty.$$
Taking $Z_n = Y_n\indicator{A}+g'_-(\hW\wedge\xsup_n)\indicator{A^c}$, one has $\E[Z_n]=1$. The Fenchel inequality further yields
\begin{align*}
	0 &\leq \hW Z_n+f(Z_n)-g(\hW) \leq 
(\hW \wedge x_n) Z_n+f(Z_n)-g(\hW)\\
&=g(\hW\wedge\xsup_n) - g(\hW) \downarrow 0, \qquad n\uparrow \infty, \qquad\text{on $A^c$}, 
\end{align*}
hence \eqref{eq:1} holds by dominated convergence. 
The case $g'_+(\hW)\indicator{A}\notin L^1$ and $g'_-(\hW)\indicator{A^c}\in L^1$ is analogous.

Suppose finally $g'_+(\hW)\indicator{A}\notin L^1$ and $g'_-(\hW)\indicator{A^c}\notin L^1$. This yields that $g$ is increasing near $\xinf$ and decreasing near $\xsup$. Without loss of generality we may assume that $g$ is nondecreasing on $A$ and nonincreasing on $A^c$. Taking $(\xsup_n)_{n \in \N}$ as above, the lack of integrability in $g'_+(\hW)\indicator{A}$ yields a nonincreasing sequence $(\xsup_n)_{n \in \N}$  by
$$\xinf_n = \inf_{x\in\R} \left\{\E\left[g'_+(\hW\vee x)\indicator{A}+
g'_-(\hW\wedge\xsup_n)\indicator{A^c}\right]\leq 1\right\}>\xinf$$
with $\lim_{n\uparrow \infty} \xinf_n=\xinf$. Thus 
$$ \E\left[g'_+(\hW\vee \xinf_n)\indicator{A}+g'_-(\hW\wedge\xsup_n)\indicator{A^c}\right]\leq 1 \leq  \E\left[g'_-(\hW\vee \xinf_n)\indicator{A}+g'_-(\hW\wedge\xsup_n)\indicator{A^c}\right]$$
by right-continuity of $g_+'$, for each $n \in \N$. 
Taking a convex combination of the random variables in the previous display, this yields $Z_n\in\partial g(\hW_n)$ with $\hW_n = \xinf_n\vee\hW\wedge\xsup_n$ with $\E[Z_n] =1$ such that
\begin{align*}
	0 &\leq \hW Z_n+f(Z_n)-g(\hW) \leq  \hW_n Z_n + f(Z_n)-g(\hW) \\
&= g(\hW_n) - g(\hW)\downarrow 0, \qquad n\uparrow \infty.
\end{align*}
We conclude again by dominated convergence.
\end{proof}

\begin{corollary}\label{C:240221}
In the setting of Theorem~\ref{T1}, the following are equivalent.
\begin{enumerate}[(i)]
\item\label{C.1} Both sides of \eqref{eq:T1} are finite.
\item\label{C.2} One of the following two conditions holds. 
\begin{enumerate}[(a)] 
\item\label{C.2.a} $1$ is in the interior of $\dom f$ and $g(W+\theta)\in L^1$. 
\item\label{C.2.b} $1\in\dom f$ is on the boundary of $\dom f$, and $W\in L^1$.
\end{enumerate}
\end{enumerate}
\end{corollary}

\begin{proof}
Taking $\eta = \theta$ in  \eqref{eq:T1} shows that $g(W+\theta)\in L^1$ is necessary for \ref{C.1} to hold. Furthermore, by the proof of Theorem~\ref{T1}, for $1$ in the interior of $\dom f$ it is also sufficient. For $1$ in $\dom f$ but not in its interior, the infimum $f(1) + \E[W]$ is finite if and only if $W\in L^1$. For $1\notin\dom f$, the infimum is trivially infinite. 
\end{proof}

\begin{remark}\label{R:240617}
The proof of Theorem~\ref{T1} reveals the following economic insight. In case~\ref{case.B}, we find ourselves in a situation where the expected marginal utility of $\hW$ can be strictly less than 1 as illustrated in Example~\ref{E:1}. The (arbitrarily large) gap in marginal utility is bridged by shifting $\hW$ downwards by at most $\frac{1}{n}$ to obtain a new wealth distribution $\hW_n$ whose marginal utility $\partial g(\hW_n)$ has an element $Z_n$ with expected value of 1. Crucially, the expected utility of $\hW_n$ is within a multiple of  $\frac{1}{n}$ to that of $\hW$. In view of the Fenchel equality $g(\hW_n)  = \hW_n Z_n +f(Z_n)$, one has that $\E[\hW Z_n +f(Z_n)]$ differs from the expected utility  of $\hW$ by no more than a multiple of $\frac{1}{n}$, which gives the desired dual result.
\end{remark}

\begin{example}\label{E:1}
Consider the utility function
$$ g(w) = 1-\e^{\e^{-w}-1}.$$
Fix $\alpha \geq 1$ and consider a probability space rich enough to support a random variable $W$ with the distribution
$$\P_W(\d x) = \e^{-x(1-\alpha)}\e^{-\e^{-x}}\indicator{x\leq 0}\d x +c_{\alpha}\delta_1(\d x).$$
Here $\delta_1$ is the Dirac delta measure concentrated on $1$ and $c_{\alpha} =
1-\int_1^\infty y^{-\alpha}e^{-y}\,\d y
\in(0,1)$ is chosen such that indeed $\P_W[\R]=1$. Observe that $c_\alpha$ is increasing in $\alpha$ with $\lim_{\alpha\uparrow\infty}c_\alpha=1$. Numerically, $c_1\approx 0.78$.
For $\alpha>1$, $\e^{-W}$ has finite exponential moment up to size $1$ but not above, which yields $\E[g(W)]>-\infty$ and $\inf \{\eta\in\R:\E[g(W+\eta )]>-\infty \}=0$.

Since
$$\P_{\e^{-W}}(\d x) =  x^{-\alpha}\e^{-x}\indicator{x\geq 1}\d x +c_{\alpha}\delta_{\e^{-1}}(\d x),$$
we have $g'_+(W)=\e^{\e^{-W}-1-W}\in L^1$ for $\alpha>2$
with  
$$\E[g'_+(W)]=\e^{-1}\left(\frac{1}{\alpha-2}+c_\alpha\e^{\e^{-1}-1}\right)<\frac{2}{5}\left(\frac{1}{\alpha-2}+1\right).$$ 
Hence, for $\alpha>3$ we have $\E[g'_+(W)]<1$ and $\hat\eta = 0$ is a corner solution, i.e.,
\[
\sup_{\eta\in\R}\{\E[g(W+\eta)]-\eta\}=\E[g(W)]. 
\]
\end{example}

The next proposition and corollary assert that in certain circumstances it is possible to manipulate wealth in such a way that the expected utility and wealth change only a little, while the expected marginal utility changes by a discrete amount. 
\begin{proposition}\label{P:240424}
Fix $f$, $g$ as in Theorem~\ref{T1}. Consider $W\in L^0$ and an event $A$ such that $g(W)\indicator{A}\in L^1$ and $g'_+(W)\indicator{A}\in L^1$. Suppose further there is $\varepsilon>0$  such that $g(W-\frac{\varepsilon}{2})\indicator{A}\notin L^1$. Then for any $s\geq 0$, there exists $Z\in L^0$ with $Z\indicator{A}\in L^1$ and  $(WZ+f(Z))\indicator{A}\in L^1$, such that 
$$\E[Z\indicator{A}] = s + \E[g'_+(W)\indicator{A}]$$ 
and
\begin{equation}\label{eq:240424}
0\leq \E[(WZ+f(Z)-g(W))\indicator{A}]\leq (s+\E[|g'_+(W)\indicator{A}|])\varepsilon.
\end{equation} 

Furthermore, there is $\tW\in L^0$ with $0\leq W-\tW\leq \varepsilon\indicator{A}$, such that $Z\in \partial g(\tW)$ on $A$,  $g(\tW)\indicator{A} = (\tW Z + f(Z))\indicator{A} \in L^1$, and
\begin{equation}\label{eq:240423b}
-(s+\E[g'_+(W)^+\indicator{A}])\varepsilon\leq \E[(g(\tW)-g(W))\indicator{A}]\leq \E[g'_+(W)^-\indicator{A}]\varepsilon.
\end{equation}
\end{proposition}

Before proving the proposition, we state and briefly justify a corollary that follows by symmetry.  

\begin{corollary} \label{C:240624}
Fix $f$, $g$ as in Theorem~\ref{T1}. Consider $W\in L^0$ and an event $A$ such that  $g(W) \indicator{A} \in L^1$ and $g'_-(W)\indicator{A} \in L^1$. Suppose further there is $\varepsilon>0$  such that $g(W+\frac{\varepsilon}{2})\indicator{A}\notin L^1$. Then for any $s\geq0$, there exists $Z\in L^0$ with $Z\indicator{A} \in L^1$ and $(WZ+f(Z))\indicator{A}\in L^1$ such that 
$$\E[Z\indicator{A}] =  \E[g'_-(W)\indicator{A}] - s,$$ 
and $0\leq \E[(WZ+f(Z)-g(W))\indicator{A}]\leq (s+\E[|g'_-(W)|\indicator{A}])\varepsilon$.
\end{corollary}
\begin{proof}
	Define the concave function $\mathring{g} = g(-\id)$ and the random variable $\mathring{W} = -W \in L^0$. Then the function $\mathring{f} = f(-\id)$ stands in  conjugate relationship with $\mathring{g}$.  Moreover, we have $\mathring{g}(\mathring{W}) = g(W)$ and $\mathring{g}'_+(\mathring{W}) = -g'_-(W)$.
	Applying now Proposition~\ref{P:240424} with $g,f$, and $W$ replaced by $\mathring{g}, \mathring{f}$, and $\mathring{W}$, respectively, yields some $\mathring{\tilde{Z}}$ and $\mathring{\widetilde{W}}$  with the stated properties.  Defining $\tilde{Z} = -\mathring{\tilde{Z}}$ and $\tW = -\mathring{\tW}$ then yields the corollary.
\end{proof}

\begin{proof}[Proof of Proposition~\ref{P:240424}]
Concavity of $g$ and $g'_+(W)\indicator{A}\in L^1$ yield $g'_+(W-\varepsilon)^-\indicator{A}\in L^1$ and 
\[
	g\left(W-\frac{\varepsilon}{2}\right)^+\indicator{A} \leq \left(g(W) - \frac{\varepsilon}{2} g^\prime_+(W)\right)^+\indicator{A} \leq |g(W)|\indicator{A} + \frac{\varepsilon}{2}  |g^\prime_+(W)|\indicator{A}  \in L^1.
\]
Using $g(W-\frac{\varepsilon}{2})\indicator{A} \notin L^1$ and concavity again we then have  
$$-\infty = \E\left[g\left(W-\frac{\varepsilon}{2}\right)\indicator{A}\right] \geq 
\E\left[g(W)\indicator{A}- \frac{\varepsilon}{2}g'_+\left(W-\frac{\varepsilon}{2}\right)\indicator{A} \right],$$ 
which yields $\E[g'_+(W-\frac{\varepsilon}{2})\indicator{A}]=\infty$, hence also $\E[g'_+(W-\varepsilon)\indicator{A}]=\infty$.
It is now also clear that $g(W-\varepsilon)^+ \indicator{A} \in L^1$ and $\E[g(W-\varepsilon)\indicator{A}] = -\infty$. 

 Define now
$$\tK = \inf \left\{K\in\R : \E\left[g'_+\left(W\indicator{\{W<K\}}
+K\vee (W-\varepsilon)\indicator{\{K\leq W\}}\right)\indicator{A}\right]\leq s + \E[g'_+(W)\indicator{A}]\right\}.$$ 
Since 
$W\indicator{\{W<K\}}+K\vee (W-\varepsilon)\indicator{\{K\leq W\}}\ \downarrow\ W-\varepsilon$ for $K\downarrow-\infty$,
monotone convergence and $\E[g'_+(W-\varepsilon)\indicator{A}]=\infty$ in turn yield that $\tK$ is finite.
Moreover, note that the function 
\begin{equation}\label{eq:240609}
K\mapsto \E\left[g'_+\left(W\indicator{\{W<K\}}
+K\vee (W-\varepsilon)\indicator{\{K\leq W\}}\right)\indicator{A}\right]
\end{equation} 
is right-continuous (because $g'_+$ is right-continuous).

Observe that 
$$ \tW=\left(W\indicator{\{W<\tK\}}+\tK\vee (W-\varepsilon)\indicator{\{\tK\leq W\}}\right)\indicator{A} + W\indicator{A^c} $$
satisfies $0\leq W-\tW\leq \varepsilon\indicator{A}$. The right-continuity of the function in \eqref{eq:240609} yields 
$$\E[g'_+(\tW)\indicator{\{\tK\leq W\}}\indicator{A}] \leq s + \E[g'_+(W)\indicator{\{\tK\leq W\}}\indicator{A}]\leq\E[g'_-(\tW)\indicator{\{\tK\leq W\}}\indicator{A}].$$ 
 We shall now discuss two subcases: 
\begin{enumerate*}
\item[\optionaldesc{($\alpha$)}{case.alpha}] $\tK > \xinf$; and 
\item[\optionaldesc{($\beta$)}{case.beta}] $\tK = \xinf>-\infty$,
\end{enumerate*}
 where $\xinf =\inf\dom g$.

\noindent \ref{case.alpha}\quad Since $g'_-$ is bounded above on $[\tK,\infty)$ and since \(g'_+(W)1_A \in L^1\), one has $g'_-(\tW)\indicator{\{\tK\leq W\}}\indicator{A}\in L^1$.   A suitable convex combination of $g'_+(\tW)\indicator{A}$ and 
$\left(g'_+(\tW)\indicator{\{W<\tK\}}+g'_-(\tW)\indicator{\{\tK\leq W\}}\right)\indicator{A}$ now yields  $Z\in L^0$ with values in $\partial g(\tW)$ on $A$ that satisfies $\E[Z\indicator{A}]= s + \E[g'_+(W)\indicator{A}]$.  

\noindent \ref{case.beta}\quad Here $\xinf=\inf\dom g>-\infty$. This and $\essinf_A ( W-\frac{\varepsilon}{2})\leq\xinf$ implies $\E[\indicator{\{\tW =\xinf\}\cap A}]>0$.  Since $g'_+(\tW) \indicator{A}\in L^1$, we have $g'_+(\xinf)$ is finite, hence 
$\partial g(\xinf) = [g'_+(\xinf),\infty)$. 
We now let
\begin{align*} 
Z ={}& g'_+(\tW)\indicator{A}+ \frac{s+\E[g'_+(W) \indicator{A}] - \E[g'_+(\tW) \indicator{A}]}{\P[\{\tW =\xinf\}\cap A]}\indicator{\{\tW=\xinf\}\cap A},
\end{align*}
which yields again $Z\in L^0$ with values in $\partial g(\tW)$ on $A$, such that $Z\indicator{A}\in L^1$ and  $\E[Z\indicator{A}]=s+\E[g'_+(W)\indicator{A}]$.   

The Fenchel inequality \eqref{Fenchel ineq} holds with equality for any $y\in\partial g(x)$ at each point $x$ where $g$ is subdifferentiable; see \cite[Theorem~23.5]{rockafellar.70}. Hence, the random variables $\tW$ and $Z$ constructed in \ref{case.alpha} and \ref{case.beta} satisfy  $g(\tW)=\tW Z + f(Z)$ on $A$. This yields
\begin{equation}\label{eq:240423c}
\begin{split}
 g(W)-\varepsilon Z^+ &{} \leq g(W)-(W-\tW)Z\\
&{}\leq W Z + f(Z)-(W-\tW)Z\\ 
&{} = \tW Z + f(Z)=g(\tW)\qquad\qquad\qquad\text{on $A$},
\end{split}
\end{equation} 
where the second line follows from the Fenchel inequality.
 
Concavity of $g$ gives
\begin{equation}\label{eq:240425b}
g(\tW) \leq g(W)+\varepsilon g'_+(W)^- \qquad\qquad\qquad\text{on $A$},
\end{equation} 
which, together with \eqref{eq:240423c}, yields $g(\tW)\indicator{A}\in L^1$. Observing that $Z^-\leq g'_+(W)^-$ on $A$ and 
$$s + \E[g'_+(W)\indicator{A}] = \E[Z^+\indicator{A}] - \E[Z^-\indicator{A}],$$ 
we obtain 
\begin{equation}\label{eq:240501}
 \E[Z^+\indicator{A}]\leq s + \E[g'_+(W)^+\indicator{A}],
\end{equation}
which yields \eqref{eq:240423b} on taking expectations in \eqref{eq:240423c} and \eqref{eq:240425b}. Fenchel inequality,  \eqref{eq:240423c}, and \eqref{eq:240425b} also give
\begin{equation*}
\begin{split}
 g(W) &{}\leq W Z + f(Z)\\ 
          &{} = g(\tW) + (W-\tW) Z\\
					&{} \leq g(W) + (W-\tW) Z + \varepsilon g'_+(W)^-\qquad\qquad\qquad\text{on $A$},
\end{split}
\end{equation*} 
which yields \eqref{eq:240424} in view of \eqref{eq:240501}.
\end{proof}
\section{Consequences of the cash-invariant representation}\label{S:3}
%
\subsection{Divergence as a measure of complete market investment opportunities.}\label{SS:3.1}
Fix a convex, lower semicontinuous function $f$ as in Theorem~\ref{T1}, with the additional property that $f\indicator{\id<0}=\infty\indicator{\id<0}$. The utility function $g$, obtained as the concave conjugate of $-f$, is then non-decreasing.
Denote the lower bound of the effective domain of $g$ by $\xinf = \inf\dom g$ and the bliss point of $g$ (the smallest argument where $g$ attains its supremum) by $\bliss=\inf \{x : g(x)=g(\infty )\}$.

Let $\Q$ be another probability measure, which we will restrict below to the natural case $\Q\ll\P$. The quantity  
$$u_\Q(x)=\sup_{{X\in L^1(\Q),\, \E_{\Q}[X]\leq0}}\{\E[g(x + X)]\},\qquad x\in\R,$$
measures the maximal expected utility under $\P$ in a statically complete market with pricing measure $\Q$ for an agent with initial wealth $x$. Observe that, trivially, $u_\Q(x)=-\infty$ for $x\notin\dom g$. The next proposition summarizes known results about $u_\Q$.
\begin{proposition} \label{P:Cierne231117}
For $f$ and $g$ as above and $\Q\ll\P$, the following are equivalent.
\begin{enumerate}[(i)]
\item\label{P:Cierne231117.i} $u_\Q(x)<g(\infty)$ for some $\xinf<x<\bliss$.
\item\label{P:Cierne231117.ii} $u_\Q(x)<g(\infty)$ for all $\xinf<x<\bliss$.
\item\label{P:Cierne231117.iii} There is $y>0$ such that $f(y\frac{\d\Q}{\d\P})\in L^1$. 
\end{enumerate} 
Furthermore, if any of the equivalent conditions \ref{P:Cierne231117.i}--\ref{P:Cierne231117.iii} holds, then 
\[u_\Q(x) = \min_{y>0} \left\{xy+\E\left[f\left(y\frac{\d \Q}{\d\P}\right)\right]\right\}, \qquad \xinf<x<\bliss.\]
\end{proposition}
\begin{proof}
See \textcite[Lemma~4.3 and Proposition~4.6]{biagini.cerny.11}. 
\end{proof}

We shall now perform an analogous exercise for the divergence utility
\begin{equation*}
V(X) = \newinf_{{ Z\in L_{+}^{1},\, \E[Z]= 1}} \{\E[XZ+f(Z)]\}.
\end{equation*}
The next result is new as far as we know.

\begin{proposition}\label{P:Cierne231118}
For a convex, lower semicontinuous function $f:\R\to(-\infty,\infty]$ such that $f\indicator{\id<0}=\infty\indicator{\id<0}$ and $\Q\ll\P$, one has 
$$\sup_{{X\in L^1(\Q),\, \E_{\Q}[X]\leq0}}\ \{V(x+X)\} = x + \E\left[f\left(\frac{\d\Q}{\d\P}\right)\right],\qquad x\in\R.$$
\end{proposition}
\begin{proof}
Theorem~\ref{T1} gives 
\begin{equation*}
\sup_{{X\in L^1(\Q),\, \E_{\Q}[X]\leq0}}\ \{V(X)\}  = \sup_{x\in\R}\,\{u_\Q(x)-x\}.
\end{equation*}
Since divergence preferences are translation-invariant, it suffices to show
\begin{equation}\label{eq:Cierne231122}
\sup_{x\in\R}\,\{u_\Q(x)-x\} = \E\left[f\left(\frac{\d\Q}{\d\P}\right)\right].
\end{equation}
 
If $f=\infty$, then $g=\infty$ and the equality holds trivially. Suppose now $f$ is proper. If $\dom f$ is a singleton, say $\{c\}$, then $g(x) = cx+f(c)$ for all $x\in\R$. One easily verifies that both sides are $\infty$ unless $c=1$ and $\Q=\P$, in which case \eqref{eq:Cierne231122} holds trivially. Finally, consider the case where $\dom f$ has non-empty interior. Then there are constants $c,k\in\R$ such that $f-f(c)-k(\id-c)$ is non-negative. Consequently, the function $v_\Q:\R\to(-\infty,\infty]$ given by 
$$v_\Q(y)= \E\left[f\left(y\frac{\d\Q}{\d\P}\right)\right]$$
is well-defined, convex, and lower semicontinuous by Fatou. The proof will be complete if we can show that $\sup_{x\in\R}\{u_\Q(x)-x\}=v_\Q(1)$.  

We shall distinguish three cases based on the effective domain of $v_\Q$. Observe that $g(\infty)$ and $\xinf$ finite yield $f(y)\leq g(\infty)-\xinf y$, $y\geq 0$, which implies $\dom v_\Q = [0,\infty)$.

\emph{Case 1:} Suppose first $\dom v_\Q=\{0\}$, which implies $g(\infty) =  f(0)$ is finite and $\xinf=-\infty$. Proposition~\ref{P:Cierne231117} now yields $u_\Q(x)=g(\infty)$, $x\leq \bliss$, hence $\sup_{x\in\R}\{u_\Q(x)-x\}=\infty = v_\Q(1)$ by sending $x$ to $-\infty$. 

\emph{Case 2:} If $\dom v_\Q$ is empty then $g(\infty)=\infty$. Furthermore, Proposition~\ref{P:Cierne231117} yields $u_\Q(x)=\infty$ for $\xinf < x < \bliss$, hence again $\sup_{x\in\R}\{u_\Q(x)-x\}=\infty = v_\Q(1)$.

\emph{Case 3:} Finally, suppose there is $y^*>0$ in $\dom v_\Q$. We claim
\begin{equation}\label{eq:Cierne231122b}
 u_\Q(x) = \inf_{y\geq0} \left\{xy+v_\Q(y)\right\}, \qquad x\in\R.
\end{equation}
Indeed, for $x>\xinf$ 
this holds as an easy consequence of Proposition~\ref{P:Cierne231117}.  
In the case that $\xinf$ finite, fix any 
$c\in\dom f$ and note that for all $z > c$ we have $(f(z)-f(c))\leq  -\xinf(z-c)$.
This yields $k\in\R$ such that $v_\Q(y)\leq k-\xinf\,y$ for all $y \geq y^*$, hence also
$$u_\Q(x) =  -\infty= \inf_{y\geq0} \left\{xy+v_\Q(y)\right\}, \qquad x<\xinf.$$
Thus $u_\Q$ equals the concave conjugate of $-v_\Q$ on $\R\setminus\{\xinf\}$. Concavity and upper semicontinuity of $u_\Q$ at $\xinf$ yields the equality \eqref{eq:Cierne231122b} everywhere. The Fenchel--Moreau theorem (e.g., \cite[Theorem~12.2]{rockafellar.70}) and the lower semicontinuity of $v_\Q$ now yield $\sup_{x\in\R}\{u_\Q(x)-xy\}=v_\Q(y)$ for all $y\in\R$. Taking $y=1$ completes the proof.
\end{proof}

\begin{remark}\label{R:Cierne231118}
By comparing Propositions~\ref{P:Cierne231117} and \ref{P:Cierne231118}, we observe that, for a given $\Q$, finiteness of the certainty equivalent of the divergence utility, i.e., $v_\Q(1)<\infty$, implies finiteness of the certainty equivalent of the expected utility, i.e., $u_\Q(x)<g(\infty)$ for all $\xinf<x<\bliss$. However, the converse fails in general since by \textcite[Lemma~5.1]{kramkov.schachermayer.99}, for arbitrary $y^*>0$, there exist $\Q$ and $f$ such that $v_\Q(y)<\infty$ only for $y\geq y^*$. Selecting, e.g., $y^*=2$, yields an example of a complete market, where the certainty equivalent of the maximal expected utility is finite but the divergence utility is infinite.
\end{remark}
%
\subsection{Canonical representation of uniformly weighted divergence utility}\label{SS:3.2}
Theorem~\ref{T1} and Corollary~\ref{C:240221} reveal the following explicit form of divergence utility over $L^\infty$. In order for $V(X)$ to be finite for \emph{all} $X\in L^\infty$, it is necessary and sufficient that $\dom g$ is unbounded and $1\in\dom f$. 
\begin{enumerate}[(A)]
\item If $1$ is an interior point of $\dom f$, then there is $\hat\eta\in\R$ such that $1\in\partial g(\hat\eta)$, which also means $-\hat\eta\in\partial f(1)$. Furthermore, Theorem~\ref{T1} yields
$$V(X) = \max_{\eta\in\R} \E[g(X+\eta)]-\eta,$$ 
hence $V(0)=g(\hat\eta)-\hat\eta = f(1)$. Thus $V$ is canonically represented by 
\begin{equation}\label{eq:Granada251230}
V(X)-V(0) =  \sup_{\eta\in\R}\,\{\E[\widetilde g(X+\eta)]-\eta \} ,\qquad X\in L^\infty,
\end{equation}
with $\widetilde g = g(\id +\hat\eta) - g(\hat\eta)$. Observe that $1\in\partial\widetilde g(0)$. This yields
\begin{align*} 
\widetilde f (y) ={}& \sup_{x\in\R}\,\{ g(x + \hat\eta) - g(\hat\eta) - (x+\hat\eta)y +\hat\eta y \}\\
 ={}& f(y) + \hat\eta (y - 1) -g(\hat\eta)+\hat\eta = f(y) - f(1) + \hat\eta (y - 1).
\end{align*}
\item If $1$ is the unique element of $\dom f$ or if $1\in\dom f$ is a boundary point of $\dom f$, then Theorem~\ref{T1} yields $V(X)=\E[X]+f(1)$ for all $X\in L^\infty$.  In this case, we obtain formula \eqref{eq:Granada251230} with $\widetilde g = \id$,  corresponding to $\widetilde f = \infty\indicator{\id\neq1}$.
\end{enumerate}
In conclusion, every divergence utility that is finite-valued over all of $L^\infty$ and normalized to $V(0)=0$ can be represented by $g$ with unbounded $\dom g$, $g(0)=0$, and $1\in\partial g(0)$. Equivalently, such utility can be represented by a divergence $f$ with $f(1)=0$, $f\geq0$, and $f$ superlinear at either $-\infty$ or $+\infty$. 
%
\subsection{Domains of monotonicity}\label{SS:3.3}

In applications, it is helpful to characterise the subset of the domain where the classical mean--variance utility $V_\MV$ coincides with its monotone modification $V_\MMV$. Here we extend these considerations to other non-monotone divergence utilities, which simultaneously yields a simplified proof for the MV--MMV case itself, cf. \citet{maccheroni.al.09}. To this end, for $p>1$ let $q = \vfrac{p}{(p-1)}$ and consider the family of functions  
\[
g_{p\M}=\frac{1-\left((1+(1-q)\id)\vee0\right)^p}{q};\qquad 
f_{p\M}=\frac{1+(p-1) \id^{q}-p\mkern1.5mu\id}{q}\indicator{\id\geq 0}+\infty\indicator{\id<0},\qquad p>1.
\]
Observe that $g_{2\M}$ and $f_{2\M}$ coincide with $g_\MMV$ and $f_\MMV$, respectively. For comparison, consider the non-monotone version of $g_{p\M}$ and its conjugate,
\[
g_{p}=\frac{1-\left|1+(1-q)\id\right|^p}{q};\qquad 
f_{p}=\frac{1+(p-1) |\id|^{q}-p\mkern1.5mu\id}{q},\qquad p>1.
\]
Here $g_{2}$ and $f_{2}$ coincide with $g_\MV$ and $f_\MV$, respectively. Denote the corresponding families of (concave) divergence utilities by $V_{p\M}$ and $V_p$, $p>1$. 
Thanks to Theorem~\ref{T1} we may choose as effective domains $\dom V_p = L^p$ and $\dom V_{p\M} = L^0_+-L^p_+$, hence 
\[\dom V_p\cap\dom V_{p\M}=L^p.\]
Furthermore, for each $W\in L^p$, Theorem~\ref{T1}  yields unique $\hat\eta_{p\M}(W)\in\R$ and $\hat\eta_p(W)\in\R$ such that 
\begin{equation}\label{eq:FOCp}
\E[g'_p(W+\hat\eta_p(W))]=1; \qquad \E[g'_{p\M}(W+\hat\eta_{p\M}(W))]=1;
\end{equation}
and 
\[
V_p(W) = \E[g_p(W+\hat\eta_p(W))]-\hat\eta_p(W); \qquad V_{p\M}(W)=\E[g_{p\M}(W+\hat\eta_{p\M}(W))]-\hat\eta_{p\M}(W).
\]
\begin{corollary}\label{C:mondom}
For $p>1$ and $W\in L^p$, the following statements hold.
\begin{enumerate}[(1)]
\item\label{1} $V_{p\M}$ is the monotone hull of $V_p$, i.e.,  
$$V_{p\M}(W)=\sup_{Y\in L^0_+}V_{p}(W-Y).$$
\item\label{2} We have $V_p(W) \leq V_{p\M}(W)$ and $\hat\eta_p(W)\leq  \hat\eta_{p\M}(W)$, with $\hat\eta_2(W) = -\E[W]$.
\item\label{3} $\hat\eta_{p\M}(W)$ is increasing in $p$.

\item\label{4} The following are equivalent.
\begin{enumerate}[(i)]
\item\label{4.i} $V_{p\M}(W)= V_p(W)$.
\item\label{4.ii} $\hat\eta_p(W)=  \hat\eta_{p\M}(W)$.
\item\label{4.iii} $W+\hat\eta_p(W)\leq p-1$.
\end{enumerate}
\end{enumerate}
\end{corollary}
Items~\ref{2} and \ref{3} are helpful when computing $\hat\eta_p(W)$ and/or $\hat\eta_{p\M}(W)$ numerically for $1<p\neq2$ as they offer one-sided bounds for the optimal shift. 
\begin{proof}
\ref{1} Observe that $g_{p\M}$ is the monotone hull of $g_{p}$. The proof then proceeds exactly as in \citet{cerny.20} for $p=2$ by arguing that the actions of taking the monotone hull and the translation-invariant hull commute since they both can be represented by infimal convolutions over $L^0$.

Alternatively, we now provide a short self-contained proof. First, note that \(g_{p\M}\) is the pointwise monotone hull of
\(g_p\). Now let \(Y\in L^0_+\). Then, for every
\(\eta\in\R\),
$
    g_p(W-Y+\eta)\leq g_{p\M}(W+\eta)$.
Taking expectations, subtracting \(\eta\), and then taking the supremum over
\(\eta\), we obtain
$    V_p(W-Y)\leq V_{p\M}(W)$.
Taking the supremum over all such \(Y\) gives one inequality. 
For the reverse inequality, fix \(\eta\in\R\) and set
$Y_\eta=(W+\eta-(p-1))^+$.
Then \(Y_\eta\in L^0_+\), \(W-Y_\eta\in L^p\), and
$g_{p\M}(W+\eta)=g_p(W-Y_\eta+\eta)$.
Therefore
\[
    \E[g_{p\M}(W+\eta)]-\eta
    =
    \E[g_p(W-Y_\eta+\eta)]-\eta
    \leq
    V_p(W-Y_\eta)
    \leq
    \sup_{\substack{Y\in L^0_+\\ W-Y\in L^p}} V_p(W-Y).
\]
Taking the supremum over \(\eta\in\R\) yields the reverse inequality. 

\ref{2} The first statement follows from \ref{1}. For the second, $g'_{p\M}\geq g'_{p}$ yields $\E[g'_{p\M}(W+\hat\eta_p(W))]\geq 1$, which shows that $\hat\eta_p(W)$ is a lower bound for $\hat\eta_{p\M}(W)$.
The last assertion follows by direct calculation with $g'_2 = 1 - \id$.

\ref{3} The assertion follows because $g'_{p\M}$ is increasing in $p$.

\ref{4} \ref{4.i} $\Rightarrow$ \ref{4.iii} The maximal domain where $g_p$ is increasing equals $(-\infty,p-1]$. If $\P[W+\hat\eta_p(W)> p-1]>0$, then 
\begin{align*}
V_p(W) = \E[g_p(W+\hat\eta_p(W))]-\hat\eta_p(W)<{}& \E[g_{p\M}(W+\hat\eta_{p}(W))]-\hat\eta_{p}(W)\\
\leq {}& \E[g_{p\M}(W+\hat\eta_{p\M}(W))]-\hat\eta_{p\M}(W) = V_{p\M}(W).
\end{align*}

\ref{4.iii} $\Rightarrow$ \ref{4.ii} Since $g'_p=g'_{p\M}$ on $(-\infty,p-1]$, we get $\E[g'_{p\M}(W+\hat\eta_{p}(W))]=1$, hence \ref{4.ii} by \eqref{eq:FOCp}.

\ref{4.ii} $\Rightarrow$ \ref{4.i} Note that \ref{4.ii} and \eqref{eq:FOCp} yield \ref{4.iii}, which in turn yields \ref{4.i}.
\end{proof}

To allow for different levels of risk aversion, we now fix $\gamma>0$ and consider utilities $V_p(\gamma\id)/\gamma$ and their monotone hulls $V_{p\M}(\gamma\id)/\gamma$. By Corollary~\ref{C:mondom}\ref{3}, the domain of monotonicity of $V_{p}(\gamma\id)/\gamma$ is precisely the set of those $W\in L^p$ that satisfy  
\[ 
\gamma W + \hat\eta_p(\gamma W)\leq  p - 1.
\]
For $p=2$, this recovers \cite[Lemma~2.1]{maccheroni.al.09} in view of $\hat\eta_2(\gamma W)=-\gamma\E[W]$. Observe, however, that for other values of $p$, one does not obtain an easily amenable expression for $\hat\eta_p(\gamma W)$, nor is it easy to trace the dependence of this expression on $\gamma$ explicitly.  
\def\MR#1{\href{http://www.ams.org/mathscinet-getitem?mr=#1}{MR#1}}
\def\ARXIV#1{\href{https://arxiv.org/pdf/#1}{arXiv:#1}}
\def\DOI#1{\href{https://doi.org/#1}{#1}}

\end{document}